\documentclass[a4paper, USenglish, thm-restate]{lipics-v2021}
\usepackage{style}

\title{Counting Triangulations of Fixed Cardinal Degrees}
\titlerunning{Counting Triangulations of Fixed Cardinal Degrees}

\author{Erin Chambers}
    {Department of Computer Science and Engineering, University of Notre Dame, United States}
    {echambe2@nd.edu}
    {https://orcid.org/0000-0001-8333-3676}
    {Research supported in part by the National Science Foundation under award CCF-2444309.}
\author{Tim Ophelders}
    {Department of Information and Computing Sciences, Utrecht University, the Netherlands \and
     Department of Mathematics and Computer Science, TU Eindhoven, the Netherlands}
    {t.a.e.ophelders@uu.nl}
    {https://orcid.org/0000-0002-9570-024X}
    {Partially supported by the Dutch Research Council (NWO) under project no. VI.Veni.212.260.}
\author{Anna Schenfisch}
    {Department of Mathematics, KTH Royal Institute of Technology, Stockholm Sweden}
    {schenf@kth.se}
    {https://orcid.org//0000-0003-2546-5333}
    {Supported by the Dutch Research Council (NWO) under project no. P21-13.}
\author{Julia Sollberger}
    {Department of Mathematics, Vrije Universiteit Amsterdam, the Netherlands}
    {j.sollberger@vu.nl}
    {}
    {}

\authorrunning{E. Chambers, T. Ophelders, A. Schenfisch, and J. Sollberger}
\Copyright{Erin Chambers, Tim Ophelders, Anna Schenfisch, and Julia Sollberger}

\hideLIPIcs
\nolinenumbers


\newcommand{\enumit}[1]{\textcolor{darkgray}{\sffamily\bfseries\upshape\mathversion{bold}#1}}

\newcommand{\sharpVC}{\textsf{\#3-regular bipartite planar vertex cover}\xspace}
\newcommand{\sharpIS}{\textsf{\#3-regular bipartite planar independent set}\xspace}
\newcommand{\sharpTiles}{\textsf{\#tiled noncrossing cycle-set}\xspace}
\newcommand{\sharpCardSig}{\textsf{\#cardinal signature realization}\xspace}

\begin{document}

\maketitle
\begin{abstract}
    A fixed set of vertices in the plane may have multiple planar straight-line triangulations in which the degree of each vertex is the same.
    As such, the degree information does not completely determine the triangulation.
    We show that even if we know, for each vertex, the number of neighbors in each of the four cardinal directions, the triangulation is not completely determined.
    In fact, we show that counting such triangulations is~\#P-hard via a reduction from \sharpVC.
\end{abstract}

\section{Introduction}
    Suppose we are given a finite set of vertices in $\R^2$ as well as the degrees of each vertex in each of the four cardinal directions. How
    many maximal triangulations (i.e., triangulations of the convex hull of the
    vertex set) satisfy these constraints?

    Of course, enumerating and counting planar triangulations and other non-crossing structures is well studied. 
    One significant result of particular relevance for this paper proves that counting the number of triangulations of a polygon with holes and integer-valued coordinates is \#P-hard~\cite{Eppstein2019,Eppstein2020}.
    Other work in this area focuses on bounding or approximating the number of triangulations, both in general and for  special classes of points~\cite{Dey1993,Epstein1994,Kaibel2003,Anclin2003,Aichholzer2004,Aichholzer2007,Aichholzer2016}, or even counting or enumerating all such triangulations in exponential time~\cite{Parvez2011,Alvarez2013,Alvarez2015}.
    Our formulation differs in that we are giving some partial information at each vertex, and asking what triangulations exist that respect this partial information, which is perhaps more akin in spirit to the partially embedded graph problem, which asks how best to extend a straight-line planar drawing of a subgraph to a planar drawing of the whole graph~\cite{Chan2014}.
    
    While quite combinatorial in nature and interesting on its own as a restricted variant of the standard counting problem, our interest in this problem comes from a very different context, motivated by the study of directional transforms in topological data analysis (TDA). Directional transforms are parameterized sets of summaries of a shape (e.g., persistence diagrams formed from sublevel set filtrations of a graph parameterized by direction), and have recently garnered significant interest in the TDA community~\cite{Turner2014,Curry2022,belton2020reconstructing,Turner2024,fasy2023faithful,fasy2024smallfaithfulsetsbe,Onus2024,Turkev2022,Chambers2025}.  In this line of work, it is well known that the Euler characteristic function or persistence diagram of a shape from height-based filtrations in many different directions allows complete reconstruction of the object, although unfortunately the best known bounds require an exponential number of directions in total.  However, these transforms taken from a small number of directions can in practice nonetheless help in many shape analysis pipelines, with a growing body of work exploring their potential applications in different data analysis pipelines~\cite{Amezquita2021,Turkev2022,Chambers2024merge,Qin2025}.  
    
    Our motivation for this paper comes from studying the inverse version of this problem, as formalized in a general sense in~\cite{Oudot2020}.  Here, we specify the inverse question as follows: if we are given the persistence diagrams of the sublevel set filtration taken from a small number of directions, how many different input shapes can generate this data? Such a study can give insights into how lossy the data is, allowing progress towards better lower bounds on the number of directions that are truly necessary.
    Indeed, we could also ask for a lower bound on the minimum number of directions required to reconstruct the object, given a particular topological signatures and class of shapes. In an effort to specify a clean combinatorial version of this problem, we note that both Euler characteristic functions and persistence diagrams come in two flavors; \emph{concise} and \emph{verbose}, depending on if they omit or include instantaneous changes to Euler characteristic or homology. In our setting of embedded graphs, the verbose Euler characteristic function or verbose persistence diagram corresponding to a single direction gives the number of edges below and adjacent to each vertex, i.e., they give indegrees with respect to that direction~\cite{belton2020reconstructing}, allowing for a clean combinatorial formulation in terms of the graph only, which we focus on for the rest of this work.

    Work in this domain is quite new, and there is very little work on understanding these lower bounds even for very simple geometric objects.  We note that~\cite{fasy2024smallfaithfulsetsbe} gives a
    construction of plane matchings for which having access to indegree information with respect to $\Omega(n)$
    directions is required to uniquely determine the edge set. With their focus on bounding the number of directions, the authors of~\cite{fasy2024smallfaithfulsetsbe} construct only \emph{pairs} of matchings with the same indegree information for many directions, and is not concerned with the problem of enumerating all such realizations.

    In order to study the computational question of enumerating realizations for
    given indegree information, we focus on graphs that satisfy given indegree
    information in each of the four cardinal directions.
    We quickly see that knowing the degrees in these directions does not
    uniquely determine the edge set. See Figure~\ref{fig:same-card-degs} for an
    example. Moreover, we show that, not only can many such
    realizations exist, but that counting the number of possible realizations
    given just cardinal-degree information is \#P-hard via a parsimonious
    reduction from \sharpVC.
    Because our constructions can be sheared to an arbitrary degree, our work implies hardness of counting realizations when we know degree information in any $d$ directions, with $d \geq 4$.
    Unfortunately, this means that even
    in~$\mathbb{R}^2$, not only is an arbitrarily large set of fixed directions generally insufficient to determine a unique input triangulation, neither is it sufficient to restrict
    or even allow efficient enumeration of the number of input triangulations.

\section{Preliminaries}
    We consider \emph{plane straight-line (PSL) graphs}: graphs whose vertices are points in the plane, and whose edges are interior-disjoint segments connecting the vertices at their endpoints.
    A PSL graph is a \emph{PSL triangulation} if all bounded faces are triangles.
    A PSL graph is \emph{maximal} if adding any edge would result in a graph that is not a PSL graph.

    For a finite set $V$ of points in the plane, let $\hull(V)$ be the cycle graph whose vertices are those of $V$ that lie on the boundary of $V$'s convex hull, and whose edges connect consecutive such vertices.
    Any maximal PSL graph $(V,E)$ is a PSL triangulation and has $\hull(V)$ as a subgraph.
    In this work, we consider only PSL graphs whose vertices have distinct $x$- and~$y$-coordinates, and whose vertices are in general position.
    
    First, we note the number of edges in a maximal PSL graph in relation to how
    many vertices lie on the boundary of its convex hull.
    
    \begin{property}\label{prop:euler}
        By Euler's formula, a PSL graph $G=(V,E)$ with $c$ vertices in $\hull(V)$ is maximal if and only if $|E|=3|V|-c-3$.
    \end{property}

    Since PSL graphs have a fixed embedding in the plane, we may use geometric
    information about edges as invariants. Rather than just the degree of a vertex, we introduce a notion of degree that takes into account where particular edges lie in relationship to the vertex.
    \begin{definition}[Cardinal Degrees]
    \label{def:cardinaldegree}
        Let $G = (V,E)$ be a PSL graph and $v \in V$. 
        The \emph{north-degree of $v$}, denoted $\deg_\dN(v)$, is the number of edges 
        incident to $v$ whose other endpoint lies strictly above $v$ in the $y$-direction 
        (conceptually ``to the north''). 
        The \emph{south-, east-, and west-degrees of $v$} ($\deg_\dS(v)$, $\deg_\dE(v)$, and $\deg_\dW(v)$) are defined analogously.
        The four functions $\deg_\dN,\deg_\dS,\deg_\dE,\deg_\dW\colon V\to\Nat$ are the \emph{cardinal degrees} of $G$.
    \end{definition}
    
    We are interested in reconstructing the edges of a PSL graph from just its vertices and cardinal degrees.
    For a vertex set $V\subset \R^2$ and functions $\deg_\dN,\deg_\dS,\deg_\dE,\deg_\dW\colon V\to\Nat$, we refer to the tuple $\sigma=(V,\deg_\dN,\deg_\dS,\deg_\dE,\deg_\dW)$ as a \emph{cardinal signature}.
    Following Definition~\ref{def:cardinaldegree}, any PSL graph $G$ has a unique cardinal signature, which we denote $\sigma_G$.
    We call $G$ a \emph{realization} of a cardinal signature $\sigma$ if and only if $G$ is a PSL graph with $\sigma_G=\sigma$.
    
    Figure \ref{fig:same-card-degs} (right) shows that a cardinal signature can have multiple realizations.
    In this work, we focus on counting how many realizations a given cardinal signature has.
    
    \begin{figure}[t]
        \centering
        \includegraphics[page=3]{same-card-degs}\hspace{1em}%
        \includegraphics[page=4]{same-card-degs}\hspace{4em}%
        \includegraphics[page=1]{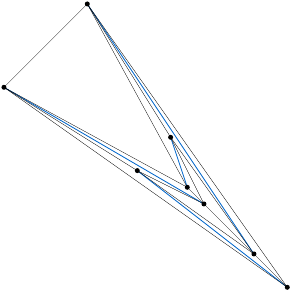}\hspace{-5em}%
        \includegraphics[page=2]{same-card-degs}%
        \caption{%
            (Left) A pair of maximal PSL triangulations with identical vertices and north- and south-degrees. (Right) A pair with identical cardinal signatures.
            Common edges are shown in black.}
        \label{fig:same-card-degs}
    \end{figure}

\section{Saturated Triangulations}
    Although we have seen that a cardinal
    signature does not necessarily have a unique realization, we are sometimes
    able to identify edges that must be present in every realization. 
    \begin{definition}[Forced Edges]
       Let $\sigma$ be a cardinal signature for a vertex set $V$. 
       For $u,v \in V$, we say that the edge~$uv$ is \emph{forced} if every possible realization of $\sigma$ contains the edge $uv$. 
    \end{definition}
    In general, we illustrate forced edges in black
    and edges that are not forced in color. Then black edges taken together
    with all edges of a particular color (or colors) form an example
    realization of some cardinal signature (see~Figure~\ref{fig:same-card-degs}).

    Let $\sigma=(V,\deg_\dN,\deg_\dS,\deg_\dE,\deg_\dW)$ be a cardinal signature.
    The number of edges in any realization of $\sigma$ is the sum of degrees in any direction.
    Specifically, for any $\mathcal{D}\in\{\dN,\dS, \dE,\dW\}$, the number of edges in any realization of $\sigma$ is $\sum_{v\in V}\deg_\mathcal{D}(v)$.
    We denote this number of edges by $m_\sigma$.
    We call the signature $\sigma$ \emph{saturated} if $m_\sigma=3|V|-c-3$, where $c$ is the number of vertices of $\hull(V)$.
    Lemma~\ref{lem:guaranteed-boundary} follows immediately from Property~\ref{prop:euler}.

    \begin{lemma}\label{lem:guaranteed-boundary}
        Any realization of a saturated signature is a maximal PSL triangulation.
    \end{lemma}
    
    \begin{figure}[b]
        \centering
        \includegraphics{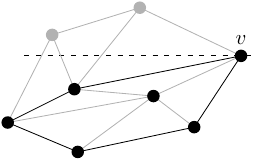}
        \caption{
            A graph whose signature saturates the halfspace below the height of $v$, thereby forcing all edges in the 
            boundary of the convex hull of this lower subgraph.
        }
        \label{fig:grid-lines}
    \end{figure}
    
    We generalize Lemma~\ref{lem:guaranteed-boundary} to certain induced subgraphs of realizations.
    Specifically, consider any closed halfspace $h$ that has a vertical or
    horizontal boundary. That is, the boundary of $h$ is either to its north or
    south, or to its east or west, respectively. For an embedded vertex set~$V$,
    we write $V_h=V\cap h$.
    For any realization $G$ of $\sigma$, the number of edges $m_h$ of $G$ contained entirely in $h$ is determined completely by $h$ and $\sigma$.
    For example, if $h$ is bounded on the north, then $m_h=\sum_{v\in V_h}\deg_\dS(v)$.
    If $h$ is bounded in a different direction, then~$m_h$ can be derived similarly from one of $\deg_\dN$, $\deg_\dE$, or $\deg_\dW$.
    We say that $\sigma$ \emph{saturates}~$h$ if~$m_h=3|V_h|-c_h-3$, where $c_h$ is the number of vertices of $\hull(V_h)$.
    Figure~\ref{fig:grid-lines} illustrates the following consequence of Lemma~\ref{lem:guaranteed-boundary}.
    
    \begin{corollary}\label{cor:guaranteed-boundary}
        If $\sigma$ saturates a closed halfspace $h$ with vertical or horizontal boundary, then for any realization $G$ of $\sigma$, the subgraph of $G$ induced by $V_h$ is a maximal PSL triangulation. In particular, all edges of $\hull(V_h)$ are forced.
    \end{corollary}

\section{\#P-Hardness of an Auxiliary Problem}\label{sec:reduce-to-tiles}
    We have already seen that some cardinal signatures have multiple realizations.
    We define \sharpCardSig to be the problem of counting the number of realizations of an input cardinal signature.
    We show this problem is \#P-hard in Sections~\ref{sec:tilings-to-cardinal-signatures} and~\ref{sec:cardinal-signatures-to-tilings}.
    We do so by a reduction from \sharpVC, which is \#P-complete as shown by M. Xia and W. Zhao~\cite{Xia2006p-complete}.
    We first reduce this problem to an intermediary problem.
    
    The input of the \sharpVC problem is a 3-regular bipartite planar graph $G=(V,E)$, and the output is the number of vertex covers of $G$. 
    Here, a \emph{vertex cover} of $G$ is any subset $V'\subseteq V$ of its vertices, for which each edge in $E$ is incident to at least one vertex in $V'$.
    An \emph{independent set} of $G$ is any subset $I$ of vertices, such that no two vertices in $I$ are connected by an edge in $E$.
    A subset $V'\subseteq V$ is a vertex cover if and only if $V\setminus V'$ is an independent set of $G$.
    Therefore, the number of independent sets of $G$ is equal to its number of vertex covers.
    As such, the corresponding \sharpIS problem is \#P-complete.

    We now introduce the auxiliary \sharpTiles problem, whose input is a \emph{tiling}:
    a $w\times h$ grid of \emph{tiles}, 
    where each tile has one of the \emph{tile types} illustrated in Figure~\ref{fig:tiles}, 
    and the boundaries of adjacent tiles match up.
    The result is a collection of red and blue cycles that may intersect.
    We call the tile types Crossing ii and Crossing iii \emph{critical}, see Figure~\ref{fig:tiles}~(h).
    Analogously, we call a tile critical if its type is critical.
    Two cycles intersect exactly in critical tiles, and in that case one cycle is blue, and the other is red.
    Any tiling has a corresponding bipartite intersection graph, whose vertices are the cycles, and whose edges connect the pairs of cycles that intersect.
    The output of the \sharpTiles problem is the number of independent sets in this intersection graph.
    
    \begin{figure}[b]
        \centering
        \includegraphics{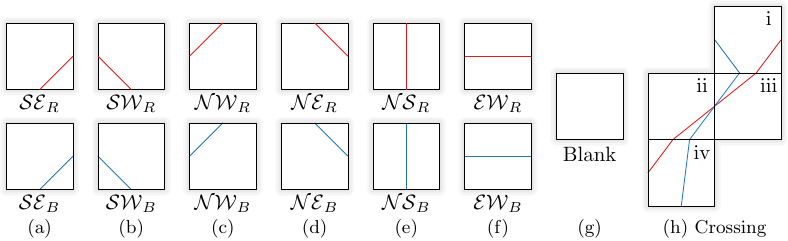}
        \caption{
            The seventeen different tile types in the \sharpTiles problem. 
            In any tiling, the four tile types in (h) always appear together as illustrated.
            In contrast to (a--f), which each have both a red and a blue tile type, the four tile types in~(h) do \emph{not} need a counterpart with colors exchanged.
        }
        \label{fig:tiles}
    \end{figure}
    
    \begin{figure}[t]
        \centering
        \includegraphics{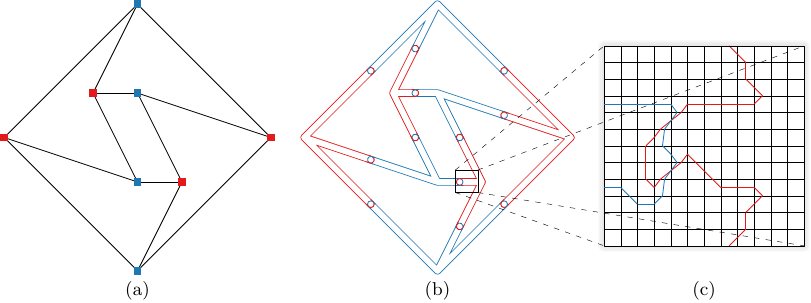}
        \caption{
            (a) An instance of \sharpIS. (b) A corresponding set of cycles in the plane. (c) Part of a corresponding tiling.
        }
        \label{fig:chicken-feet}
    \end{figure}
    
    Next, we establish hardness of the \sharpTiles problem.
    
    \begin{lemma}
    \sharpTiles is \#P-hard.
    \end{lemma}

    \begin{proof}
    We proceed via reduction from \sharpIS.
    Consider any instance $G=(V,E)$ of \sharpIS.
    Let $R\sqcup B=V$ be a bipartition, such that each edge $e\in E$ connects a vertex in $R$ to a vertex in~$B$.
    Such a bipartition can be computed in linear time using standard graph search algorithms.
    We can realize $G$ as the intersection graph of a set of simple closed cycles $\{C_v\}_{v\in V}$ in the plane, such that any two cycles intersect in at most two points, and each intersection is transversal, as illustrated in Figure~\ref{fig:chicken-feet}~(a--b).
    We may moreover assume that every cycle lies on a polynomial-sized grid~\cite{Tamassia1989}, and is constructed using the tile types of Figure~\ref{fig:tiles} so as to form a tiling $X$, whose corresponding intersection graph is~$G$, and for which the cycle $C_v$ is red if $v\in R$, and blue if $v\in B$, see Figure~\ref{fig:chicken-feet}~(c).
    Because~$G$ is the intersection graph of the cycles of $X$, and~$X$ was constructed in polynomial time from $G$, the result follows.
    \end{proof}
    
    Consider a tiling $X$.
    Each tile of $X$ contains at most one subpath of at most one red cycle, and at most one subpath of at most one blue cycle.
    Let $X_R$ be the subset of tiles that contain part of a red cycle, and $X_B$ be the subset of tiles that contain part of a blue cycle.
    A \emph{tile selection} is a pair $(S_R,S_B)$ with $S_R\subseteq X_R$ and $S_B\subseteq X_B$.
    We say a tile selection~$(S_R,S_B)$ is \emph{noncrossing} if it meets the following constraints:
    \begin{enumerate}
        \item If a red cycle passes through the common boundary of two adjacent tiles $t$ and $t'$, then~$t\in S_R$ if and only if $t'\in S_R$.
        \label{constr:red-signal}
        \item If a blue cycle passes through the common boundary of two adjacent tiles $t$ and $t'$, then~$t\in S_B$ if and only if $t'\in S_B$.
        \label{constr:blue-signal}
        \item If $t$ is a critical tile (i.e., of type Crossing~ii or~iii), then $t\notin S_R$ or $t\notin S_B$.\label{constr:noncrossing}
    \end{enumerate}
    
    \noindent
    Lemma~\ref{lem:counting-noncrossing-cycle-sets} shows that \sharpTiles is the same as counting the number of noncrossing tile selections.
    \begin{restatable}{lemma}{counting-noncrossing-cycle-sets}
    \label{lem:counting-noncrossing-cycle-sets}
        The noncrossing tile selections of any tiling $X$ are in bijection with the independent sets in the intersection graph of $X$.
    \end{restatable}
    \begin{proof}
        Consider a noncrossing tile selection $(S_R,S_B)$, and let $I$ be the set consisting of the red cycles that visit a tile in $S_R$, and the blue cycles that visit a tile in $S_B$.
        
        Any pair of intersecting cycles crosses in a critical tile, and hence contains one red and one blue cycle.
        If a red cycle $C_R$ intersects a blue cycle $C_B$ in a critical tile $t$, then by Constraint~\ref{constr:noncrossing} either $t\notin S_R$ or $t\notin S_B$.
        If $t\notin S_R$, then by transitivity of Constraint~\ref{constr:red-signal}, $C_R$ does not pass through any tile in $S_R$, and hence $C_R$ does not lie in $I$.
        Symmetrically based on Constraint~\ref{constr:blue-signal}, if $t\notin S_B$, then $C_B$ does not lie in $I$.
        Hence, for any pair of intersecting cycles, at least one of them does not lie in $I$, which is hence an independent set.
        
        For the reverse direction, any independent set $I$ in the intersection graph of $X$ corresponds to a noncrossing tile selection $(S_R,S_B)$, in which a tile $t\in X_R$ (respectively $X_B$) lies in $S_R$ (respectively $S_B$) if and only if the red (respectively blue) cycle that passes through $t$ lies in~$I$.
        The above two constructions are each others' inverses, and hence form a bijection.
    \end{proof}

\section{Tilings to Cardinal Signatures}\label{sec:tilings-to-cardinal-signatures}
    Let $X$ be a tiling with $w\times h$ tiles.
    In this section, we describe the construction of a cardinal signature $\sigma$, and show that noncrossing cycle sets of $X$ are in bijection with a certain class~$\mathcal{R}$ of realizations of~$\sigma$.
    In Section~\ref{sec:cardinal-signatures-to-tilings}, we establish the \#P-hardness of \sharpCardSig by showing that every realization of~$\sigma$ lies in the class $\mathcal{R}$.
    
    \subsection{Constructing the Cardinal Signature}
        \subparagraph*{Frames for Realizations}
        We first introduce the vertices of $\sigma$, which lie near the edges of a~$w\times h$ grid of squares.
        In $\sigma$, each tile will correspond to a frame cell.
        
        \begin{figure}[b]
            \centering
            \includegraphics[page=2]{frame_cell}
            \hspace{3em}
            \includegraphics{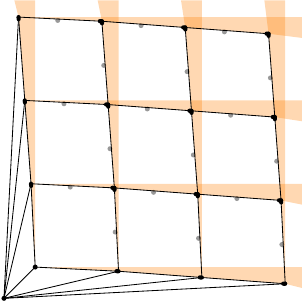}
            \caption{(Left) A frame cell, which can be sheared to be arbitrarily close to square while retaining the specification of Construction~\ref{cons:framecell}. (Right) A $3\times 3$ frame graph, as in Construction~\ref{cons:framegraph}. The regions that are forbidden from containing vertices in their interiors are indicated by shaded orange.}
            \label{fig:framecell}
        \end{figure}
        
        \begin{construction}[Frame Cell]
        \label{cons:framecell}
            See Figure~\ref{fig:framecell} (left).
            A \emph{frame cell} consists of a PSL cycle with 16 vertices, denoted in
            counterclockwise order by $v_0, \vb_{1}, \vb_{2}, \vb_{3}, \vb_{4},
            \vb_{5}, \vb_{6}, \vb_{7}, v_{8},$ $v_{7}, v_{6}, v_{5}, v_{4}, v_{3},
            v_{2}$ and $v_1$.
            A frame cell has no additional edges, but, depending on the tile type, has up to
            two interior vertices, denoted $p$ and $q$.
            The order of vertices in increasing $x$-coordinate is
            \begin{align*}
                 \vb_4, \vb_5, \vb_6, \vb_7, v_8, p, \vb_3,\vb_2,\vb_1, v_0, v_1,v_2,v_3, q, v_7,v_6,v_5,v_4,
            \end{align*}
            or the same list without $p$ or $q$. 
            The order of vertices in increasing $y$-coordinate is
            \begin{align*}
                v_4,v_5,v_6,v_7,v_8, q, v_3,v_2,v_1,v_0,\vb_1,\vb_2,\vb_3, p, \vb_7,\vb_6,\vb_5,\vb_4,
            \end{align*}
            or the same list without $p$ or $q$. Finally, if present, we ensure $p$ is to the south and west
            of the segment $\vb_3\vb_7$ and we ensure $q$ is to the south and west of the segment $v_3v_7$. 
            We call the paths defined by $v_0,v_1,v_2,v_3,v_4$, by $\vb_4, \vb_5, \vb_6,\vb_7,v_8$,
            by $v_0, \vb_1, \vb_2, \vb_3, \vb_4$, and by $v_4,v_5,v_6,v_7,v_8$ the \emph{east, west, north,} and \emph{south boundaries}, respectively.
        \end{construction}
    
        Sufficiently ``square'' frame cells can be arranged together to form a grid-like structure.
        With the goal of forcing edges in this construction, we introduce additional restrictions on how frame cells are placed
        together in Construction~\ref{cons:framegraph}, which is illustrated in Figure~\ref{fig:framecell} (right).
    
        \begin{construction}[Frame Graph]\label{cons:framegraph}
            First, arrange $(w+1)(h+1)$ vertices in a grid, connected by edges to form a lattice of $w \times h$ squares. Denote the vertical path that is $i$-th from the left by $\alpha_i$, and the horizontal path that is $j$-th from the bottom by $\beta_j$.
            Perturb the corners of these squares so that vertices previously in a vertical path $\alpha$ now form a concave-left path $\Tilde{\alpha}_i$, and vertices previously in a horizontal path $\beta$ now form a concave-down path $\Tilde{\beta}_j$.
            
        Finally, we subdivide every edge of every path $\Tilde{\alpha}_i$ and $\Tilde{\beta}_j$ (i.e., every edge of a square in our grid) with three internal vertices, and perturb them so that the resulting path forms a convex chain, 
            and the resulting square-like cycles are frame cells.%
            \footnote{Note that vertices may belong to up to four frame cells, and thus may have multiple names; 
            we disambiguate this by specifying with respect to which frame cell we discuss a vertex.} Let $\Tilde{\alpha}'_i$ and $\Tilde{\beta}'_j$ denote the resulting ``nearly vertical'' and ``nearly horizontal'' subdivided paths, respectively.
            As described in Construction~\ref{cons:framecell}, we also add interior vertices~$p$ or $q$ to frame cells depending on the type of the corresponding tile.
            We additionally ensure that their placement satisfies the following.
            \begin{enumerate}
                \item For each $0 \leq i \leq w$, the path $\Tilde{\alpha}'_i$ lies on the convex hull of all the vertices west of the southmost vertex of $\Tilde{\alpha}'_i$.
                \item For each $0 \leq j \leq h$, the path $\Tilde{\beta}_j$ lies on the convex hull of all the vertices south of the westmost vertex of $\Tilde{\beta}'_j$.
            \end{enumerate} 
            The above conditions cause particular regions to contain no vertices in their interior, and these regions are illustrated by shaded orange in Figure~\ref{fig:framecell}.
            Finally, we add a single \emph{cone vertex $c$} to the south and west of
            all frame cells in this structure, and add edges from $c$ to all
            vertices in $\Tilde{\alpha}'_0$ (i.e., to all vertices on west boundaries of frame cells in the first column) and to
            all vertices in $\Tilde{\beta}'_0$ (i.e., to all vertices on south boundaries of frame cells in the last row).
            
            We call this construction a \emph{frame graph}, denoted~$F$.
        \end{construction}
        
        Because the nearly horizontal and vertical paths in a frame graph are convex, we have the following desirable property.
        \begin{lemma}[Forced Frame Edges]
        \label{lem:forceframe}
            Suppose that $G=(V,E)$ is a maximal PSL triangulation with cardinal signature
            $\sigma$. Also suppose that the frame graph $F$ is a
            subgraph of $G$, and~$G$ contains no vertices other than those of $F$. 
            Then the edges of $F$ appear in every realization of $\sigma$; that is, the edges of $F$ are forced.
        \end{lemma}
        \begin{proof}
            Let $h$ denote the halfspace below the $x$-coordinate of the west-most
            vertex of some nearly horizontal path $\tilde{\beta}$ of $F$. By construction, $\tilde{\beta}$ is on the boundary of $\hull(V_h)$. Since,
            additionally,~$G$ is a maximal PSL triangulation, $\sigma$ saturates
            $h$. Then by Corollary~\ref{cor:guaranteed-boundary}, all edges in $\tilde{\beta}$ are forced. A nearly identical argument holds if
            we consider the halfspace to the west of the $y$-coordinate of the
            north-most vertex in a nearly vertical path of $F$. Next, consider
            the cone vertex of $F$, denoted $c$. The north degree of $c$
            equals the number of vertices on the south and west boundaries of the
            frame cells in the first column and last row of $F$, respectively.
            Since these boundaries are instances of nearly horizontal and vertical
            paths as discussed previously, we know $c$ cannot be adjacent to any
            other vertices to its north, otherwise, we would have an edge crossing.
            Thus, the edges adjacent to $c$ in $F$ are forced.
            We have shown that each edge of $F$ is forced.
        \end{proof}

        \begin{figure}[p!]
            \centering
            \includegraphics{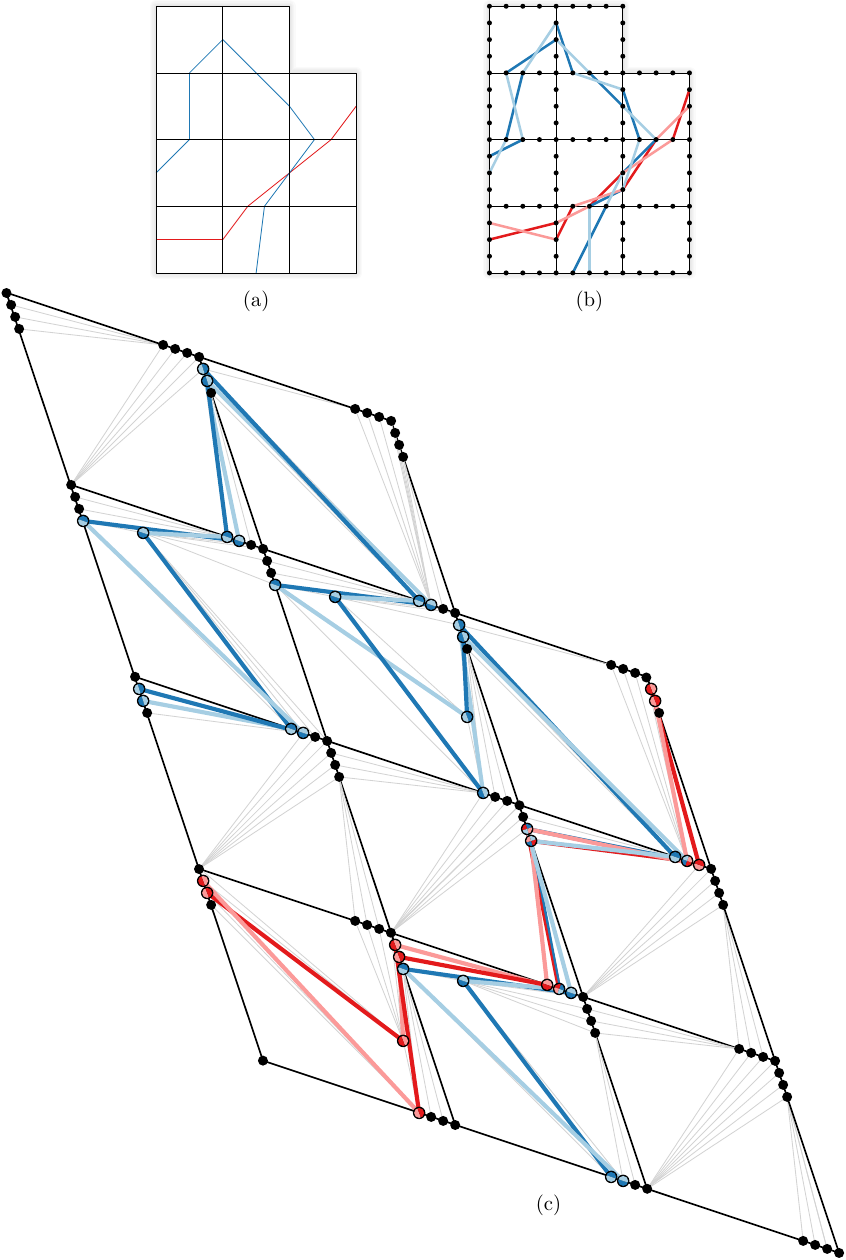}%
            \caption{(a) All possible types of tiles for the \sharpTiles problem 
            (up to simultaneously exchanging colors and shades of edges in monochromatic types of gadgets). 
            (b) A schematic representation of the gadgets of (c).
            (c) The corresponding gadgets in the same relative positions for the \sharpCardSig
            problem. To stress the position of blue and red wires on the right, we include black edges only with a thin faded line,
            other than those defining frame cells, which are drawn in solid black.}
            \label{fig:grid_gadgets}
        \end{figure}
        
    \subparagraph*{Filling the Frame Cells}
        Next, we introduce the gadgets that correspond to each tile type.
        First, we specify the data that constitutes a gadget.
        \begin{definition}[Gadget]
            A gadget $\dG=(V_\dG,R_\dG)$ consists of a set of vertices $V_\dG$ of some frame cell and a set~$R_\dG$ of \emph{intended realizations} of $V_\dG$.
        \end{definition}
        
        We illustrate the intended realizations of a gadget $\dG$ of a specific tile type via Figure~\ref{fig:grid_gadgets}, which displays an edge-colored graph $U_\dG$ that is a supergraph of each intended realization. 
        When we discuss a general gadget, we always mean one of the gadgets shown in Figure~\ref{fig:grid_gadgets}.
        In order to specify the intended realizations, we define a \emph{local coloring} to consist of three colors: black, and either light or dark red, and either light or dark blue.
        We denote the local colorings 
        $\chi[\lR\lB]=\{\text{black},\text{light red},\text{light blue}\}$, 
        $\chi[\lR\dB]=\{\text{black},\text{light red},\text{dark blue}\}$, 
        $\chi[\dR\lB]=\{\text{black},\text{dark red},\text{light blue}\}$, and 
        $\chi[\dR\dB]=\{\text{black},\text{dark red},\text{dark blue}\}$.
        For a local coloring $\chi$, denote by $U_\dG^\chi$ the subgraph of $U_\dG$ consisting of the edges of the local coloring.
        The graph $U_\dG^\chi$ is a triangulation of $V_\dG$ unless it has crossing edges.
        Observe that for a local coloring $\chi$, the edges of $U_\dG^\chi$ cross if
        and only if $\chi=\chi[\lR \lB]$ and $\dG$ corresponds to a critical tile type.
        The set $R_\dG$ of intended triangulations consists of the (at most four) graphs~$U_\dG^\chi$ without crossing edges, where $\chi$ ranges over local colorings.
        Note that $\bigcup R_\dG=\bigcup U_\dG^\chi=U_\dG$.

    \subparagraph{The Cardinal Signature}\label{sec:frame-signature}
        To define the cardinal signature $\sigma$ corresponding to the tiling $X$, 
        we need to specify both a vertex set~$V$, and the cardinal degrees for each vertex.
        The vertex set $V$ of $\sigma$ consists of the vertices of a frame graph with the same number and arrangement of frame cells as $X$ has tiles, and with copies of $p$ and $q$ whenever a tile type corresponds to a gadget with these interior vertices.
        We now specify the cardinal degrees.
        
        In the frame graph, we insert in each frame cell, the graph $U_\dG$ of
        its corresponding gadget (as in Figure~\ref{fig:grid_gadgets}).
        In each cardinal direction, any vertex of the resulting graph has as many light red edges as it has dark red edges, and as many light blue edges as it has dark blue edges.
        Each vertex also has at most one edge of each of the colors dark red, light red, dark blue, and light blue.
        As we intend for at most one blue and one red edge to be selected per vertex, we set the cardinal degree of a vertex $v$ in any particular cardinal direction to be its number of incident black, plus half the number of incident colored edges in that direction.

    \subsection{Mapping Noncrossing Tile Selections to Intended Realizations}
        We want to match every noncrossing tile selection of $X$ to a unique realization of $\sigma$.
        We do this by first defining a map from tiles to local colorings, and then to the intended realization of their associated gadget with this local coloring.
        
        Let $\varphi$ be a function that assigns a local coloring to each tile of $X$.
        Let $F(\varphi)$ be the graph that is the union of the frame graph and, for each tile $t$ with corresponding gadget $\dG$, the realization $U_\dG^{\varphi(t)}$.
        The following constraints on $\varphi$ reflect Constraints~\ref{constr:red-signal}--\ref{constr:noncrossing} of Section~\ref{sec:reduce-to-tiles}:
        \begin{enumerate}
            \item If a red cycle passes through the common boundary of two adjacent tiles $t$ and $t'$, then $\varphi(t)$ and $\varphi(t')$ contain the same shade of red.
            \label{constr:phi-red-signal}
            \item If a blue cycle passes through the common boundary of two adjacent tiles $t$ and $t'$, then $\varphi(t)$ and $\varphi(t')$ contain the same shade of blue.
            \label{constr:phi-blue-signal}
            \item If $t$ is a critical tile, then $\varphi(t)$ is not $\chi[\lR\lB]$.\label{constr:phi-noncrossing}
        \end{enumerate}
        
        \begin{observation}\label{obs:phi}
            Let $\varphi$ be a function that assigns a local coloring to each tile of $X$.
            \begin{itemize}
                \item $F(\varphi)$ satisfies the cardinal degrees of $\sigma$ if and only if $\varphi$ meets Constraints~\ref{constr:phi-red-signal} and~\ref{constr:phi-blue-signal}.
                \item $F(\varphi)$ is planar (and a maximal PSL triangulation) if and only if $\varphi$ meets Constraint~\ref{constr:phi-noncrossing}.
                \item $F(\varphi)$ is a realization of $\sigma$ if and only if $\varphi$ satisfies Constraints~\ref{constr:phi-red-signal}--\ref{constr:phi-noncrossing} simultaneously.
            \end{itemize}
        \end{observation}
        
        The constant map $\varphi$ that assigns $\chi[\dR\dB]$ meets Constraints~\ref{constr:phi-red-signal}--\ref{constr:phi-noncrossing} simultaneously.
        For this particular map $\varphi$, $F(\varphi)$ by Observation~\ref{obs:phi} meets the conditions of Lemma~\ref{lem:forceframe}, so we obtain:
        \begin{lemma}\label{lem:forceframe2}
            Any realization of $\sigma$ contains the frame graph.
        \end{lemma}
        
        Denote by $\mathcal{R}$ the set of realizations of $\sigma$ that contain for every frame cell, an intended realization of its gadget.
        Because each intended realization of a gadget $\dG$ is of the form $U_\dG^\chi$ for some local coloring $\chi$, we have the following:
        \begin{lemma}\label{lem:phi-F}
            $\mathcal{R}$ consists of exactly the graphs $F(\varphi)$ for which $\varphi$ meets Constraints~\ref{constr:phi-red-signal}--\ref{constr:phi-noncrossing}.
        \end{lemma}

        For any tile selection of $X$, we construct a 
        map $\varphi$ and graph $F(\varphi)$.
        We show that this construction yields a bijection between noncrossing tile selections and realizations in~$\mathcal{R}$.

        \begin{theorem}\label{thm:tile-selection-bijection}
            Let $(S_R, S_B)$ be a tile selection of $X$.
            Let $\varphi_{S_R, S_B}$ be the function that maps any tile $t$ to the local coloring that consists of the colors
            \begin{itemize}
                \item black;
                \item light red if $t \in S_R$, or dark red if $t \not\in S_R$; and
                \item light blue if $t \in S_B$, or dark blue if $t \not\in S_B$.
            \end{itemize}
            The map $(S_R,S_B)\mapsto F(\varphi_{S_R,S_B})$ is a bijection between noncrossing tile selections of $X$ and~$\mathcal{R}$.
        \end{theorem}
        \begin{proof}
            Let $(S_R,S_B)$ be a noncrossing tile selection of $X$.
            By construction, $\varphi_{S_R,S_B}$ meets Constraints~\ref{constr:phi-red-signal}--\ref{constr:phi-noncrossing}, so $F(\varphi_{S_R,S_B})\in\mathcal{R}$. It remains to show injectivity and surjectivity.

            For injectivity, suppose that $(S_R^1,S_B^1)$ and $(S^2_R,S^2_B)$ are distinct noncrossing tile selections of $X$.
            Then there exists a tile $t$ such that (a) exactly one of $S^1_R$ and $S^2_R$ contains $t$, or (b) exactly one of $S^1_B$ and $S^2_B$ contains $t$.
            In case (a), we have $t\in X_R$, and the realizations of its gadget will differ between $F(\varphi_{S^1_R,S^1_B})$ and $F(\varphi_{S^2_R,S^2_B})$.
            The argument for case (b) is symmetric.
            So, distinct noncrossing tile selections map to distinct graphs, proving injectivity.
            
            For surjectivity, consider an arbitrary element of $\mathcal{R}$.
            By Lemma~\ref{lem:phi-F}, it can be expressed as $F(\varphi)$ for some such function $\varphi$ that meets Constraints~\ref{constr:phi-red-signal}--\ref{constr:phi-noncrossing}.
            Based on $\varphi$, we construct some noncrossing cycle set $(S_R,S_B)$ of $X$ such that $F(\varphi_{S_R,S_B})=F(\varphi)$.
            Let $S_R=\{t\in X_R\mid \text{light red}\in\varphi(t)\}$ and $S_B=\{t\in X_B\mid \text{light blue}\in\varphi(t)\}$.
            Because $\varphi$ meets Constraints~\ref{constr:phi-red-signal}--\ref{constr:phi-noncrossing} from the beginning of this section, $(S_R,S_B)$ is indeed a noncrossing tile selection.
            Now suppose for a contradiction that ${F(\varphi_{S_R,S_B})\neq F(\varphi)}$.
            Then there exists some tile~$t$ such that~$U_\dG^{\varphi_{S_R,S_B}(t)}\neq U_\dG^{\varphi(t)}$.
            Among $U_\dG^{\varphi_{S_R,S_B}(t)}$ and $U_\dG^{\varphi(t)}$, either (a) exactly one of them uses light red edges, or (b) exactly one of them uses light blue edges.
            In case (a), $\dG$ has a realization with a light red edge, so $t\in X_R$.
            Because $t\in X_R$, $\varphi_{S_R,S_B}(t)$ by construction contains light red if and only if $\varphi(t)$ contains light red.
            Therefore, we are not in case (a).
            A symmetric argument excludes case~(b), which is a contradiction, and proves surjectivity.
        \end{proof}

        At this point, we have shown that the noncrossing tile selections are in bijection with the realizations of $\sigma$ in which each gadget uses one of the intended realizations.
        In the next section, we show that these realizations are the only ones possible.
        Specifically, we show that there is no realization of $\sigma$ in which a gadget does not use any of its intended realizations.

\section{Mapping Noncrossing Tile Selections to Realizations: Surjectivity}\label{sec:cardinal-signatures-to-tilings}
    In Section~\ref{sec:tilings-to-cardinal-signatures}, we constructed a cardinal signature $\sigma$ that admits a realization $F(\varphi_{S_R,S_B})$ for every noncrossing tile selection $(S_R, S_B)$ of a tiling $X$.
    In this section, we show that these realizations are the only realizations of $\sigma$.
    
    From Lemma~\ref{lem:forceframe2}, we know that every realization of $\sigma$
    contains the frame graph. The main objective of this section is to show that
    each frame cell admits only the intended realizations of its gadget.
    To this end, we consider \emph{triangulations of a gadget~$\dG$}: triangulations with vertex set $V_\dG$ that have the frame cell boundary on their outer face.
    To make our aim precise, we formulate a notion of correctness for a gadget in Definition~\ref{def:correct_gadget}.
    
    The possible realizations of a gadget for a fixed cardinal signature depend on the realizations of its neighbors.
    Under the assumption that each gadget is correct, in Section~\ref{sec:induction}, we show for a gadget $\dG$ that, if the neighboring gadgets on the south and west admit only the intended realizations, then so does $\dG$.
    As such, we may assume that the gadgets to the south and west correspond to (possibly distinct) local colorings.
    From these local colorings, we can infer how many edges each vertex on the south and west boundary of $\dG$ must be used to triangulate its frame cell.
    The vertices on the north and east boundaries are less constrained.
    Correspondingly, we for each gadget $\dG$ associate a set of \emph{constrained} cardinal~directions for each vertex.
    
    \begin{figure}[b]
        \centering
        \includegraphics[page=1]{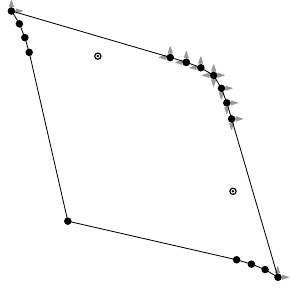}\hspace{3em}
        \includegraphics[page=2]{unconstrained}
        \caption{A frame cell with arrows indicating the unconstrained (left) and constrained (right) cardinal directions of each vertex.}
        \label{fig:relevant_directions}
    \end{figure}
    
    \begin{definition}[Constrained Directions for Vertices of $V_\dG$]
        See also Figure~\ref{fig:relevant_directions}.
        We say that a cardinal direction is \emph{constrained} for a vertex $v$ if it is not unconstrained for $v$, and define the constrained directions via the unconstrained directions.
        \begin{itemize}
            \item The north direction is unconstrained for the vertices $v_0$, $\vb_1$, $\vb_2$, $\vb_3$, $\vb_4$, and $v_4$ of $V_\dG$; 
            that is, the vertices on its north boundary, as well as the south east corner.
            \item The east direction is unconstrained for the vertices $v_0$, $v_1$, $v_2$, $v_3$, $v_4$, and $\vb_4$ of $V_\dG$; 
            that is, the vertices on its east boundary, as well as the north west corner.
            \item The south direction is unconstrained for the vertices $v_0$, $v_1$, $v_2$, and $v_3$ of $V_\dG$.
            \item The west direction is unconstrained for the vertices $v_0$, $\vb_1$, $\vb_2$, and $\vb_3$ of $V_\dG$.
        \end{itemize}
    \end{definition}

    We now define what it means for a vertex to match a local coloring.

    \begin{definition}[Matching a Local Coloring]
        Consider a triangulation of a gadget $\dG$.
        We say that a vertex $v\in V_\dG$ \emph{matches a local coloring $\chi$ in a cardinal direction} if its degree in that direction equals that of $U_\dG^\chi$.
    \end{definition}

    Using this terminology, we define the valid triangulations and correctness of a gadget~$\dG$.
    
    \begin{definition}[Valid triangulation of $\dG$]
    \label{def:valid_triangulation}
        A triangulation of the gadget $\dG$ is \emph{valid} if
        \begin{enumerate}
            \item for every vertex $v\in V_\dG$, there exists a local coloring $\chi$ such that $v$ matches $\chi$ in all its constrained directions;\label{cond:valid_vertex}
            \item there exists a local coloring $\chi_\dW$ such that the vertices on the west boundary of $\dG$ ($\vb_4$, $\vb_5$, $\vb_6$, $\vb_7$, and $v_8$) match $\chi_\dW$ in their constrained directions, and;\label{cond:valid_west}
            \item there exists a local coloring $\chi_\dS$ such that the vertices on the south boundary of $\dG$ ($v_4$, $v_5$, $v_6$, $v_7$, and $v_8$) match $\chi_\dS$ in their constrained directions.\label{cond:valid_south}
        \end{enumerate}
    \end{definition}

    \begin{definition}[Correctness of a Frame Cell Gadget]
    \label{def:correct_gadget}
        We say that a gadget $(V_\dG,R_\dG)$ is \emph{correct} if its valid triangulations are exactly the intended triangulations in $R_\dG$.
    \end{definition}

    In Section~\ref{sec:gadget-correctness}, we show that the gadget introduced for each tile type is correct.
    However, we first assume our gadgets are correct, and show in Section~\ref{sec:induction} how this lets us determine the number of cardinal signature realizations.

\subsection{Inductive Argument}\label{sec:induction}
        Assuming that each gadget is correct as specified by Definition~\ref{def:correct_gadget}, we now show that $(V,\sigma)$ does not admit cardinal signature realizations other than those that uniquely correspond to noncrossing tile selections.
        Lemma~\ref{lem:counting-noncrossing-cycle-sets} then implies that the number of noncrossing cycle-sets of $X$ equals the number of cardinal signature realizations of $(V,\sigma)$.
        
        Let $T$ be an arbitrary cardinal signature realization of $\sigma$.
        By Lemma~\ref{lem:forceframe2}, $T$ is a maximal PSL triangulation that contains a frame graph, 
        so it makes sense to discuss the triangulation of specific frame cells and corresponding gadgets.
        We denote by $\dG_{i,j}$ the gadget in column $i$ from the left and row $j$ from the bottom, and by $T_{i,j}\subseteq T$ its triangulation in $T$.

        Our goal is to show that each $T_{i,j}$ is a valid triangulation of $\dG_{i,j}$ under Definition~\ref{def:valid_triangulation}.
        If we assume that all gadgets are correct, this means that $T_{i,j}$ will be an intended realization of~$\dG_{i,j}$.
        To reason about the vertices on the south and west boundaries of $\dG_{i,j}$, 
        we will assume that for all gadgets $\dG_{i',j'}$ in the south-west directions (with $i'\leq i$ and $j'\leq j$) that $T_{i',j'}$ is an intended triangulation.
        This assumption will be justified via a north-eastward induction on frame cells.
        To show validity of $T_{i,j}$, we will determine matching local colorings for each vertex of $\dG_{i,j}$.
        We start with vertices that do not rely on induction.
        
        \begin{lemma}[Matching Local Coloring for Vertices not on the South or West Boundary]
        \label{lem:interiornorthandeast}
            Let~$T$ be a cardinal signature realization of $\sigma$ and $T_{i,j}\subseteq T$ be its triangulation of gadget $\dG_{i,j}$.
            Let~$v\in V_{\dG_{i,j}}$ be a vertex that does not lie on the south or west boundary of $\dG_{i,j}$.
            Let $\chi$ be any local coloring.
            Then in $T_{i,j}$, vertex $v$ matches $\chi$ in each constrained direction $\mathcal{D}$ of $v$.
        \end{lemma}
        \begin{proof}
            We claim that for any edge $e=\{v,u\}$ of $T$ incident to $v$ in the constrained direction~$\mathcal{D}$,
            the other end point $u$ must lie in $V_{\dG_{i,j}}$.
            Indeed, to connect to any vertex not in $V_{\dG_{i,j}}$, the edge~$e$ would have to cross the frame graph, resulting in a non-plane graph.
            Because $T$ is plane, $u$ must therefore lie in $V_{\dG_{i,j}}$.
            For $e$ to lie outside the frame cell of $\dG_{i,j}$, it needs to connect two vertices on its south or west boundary, but $v$ does not lie on the south or west boundary.
            Therefore, any edge $e$ incident to $v$ in direction $\mathcal{D}$ lies in $T_{i,j}$.
            Thus, the cardinal signature of~$T_{i,j}$, which we denote $\sigma_{i,j}$, agrees with $\sigma$ on these cardinal degrees of $v$ in direction~$\mathcal{D}$.
            
            The cardinal degree of $v$ in~$U_\dG^\chi$ is invariant under the choice of the local coloring $\chi$, 
            and equals the cardinal degree prescribed by $\sigma_{i,j}$.
        \end{proof}

        Next, we consider vertices on the west boundary and show that in $T_{i,j}$, the cardinal degrees match the local coloring of the frame cell to the west (or any local coloring if no such cell exists).
        Lemma~\ref{lem:west} establishes Condition~\ref{cond:valid_west} by reasoning about the ``leftover'' degree of a vertex that a realization must use inside a frame cell based on how many incident edges are already used in frame cells to the south and west.
        
        \begin{lemma}[Condition~\ref{cond:valid_west} Assuming Intended Triangulations to the South-West]
        \label{lem:west}
            Let~$T$ be a cardinal signature realization of $\sigma$ and $T_{i,j}\subseteq T$ be its triangulation of gadget $\dG_{i,j}$.
            Fix some~$i$ and $j$ and assume for all $(i',j')\neq (i,j)$ with $i'\leq i$ and $j'\leq j$, that there exists a local coloring $\chi_{i',j'}$ such that $T_{i',j'}=U_{\dG_{i',j'}}^{\chi_{i',j'}}$.

            If $i=0$, let $\chi_\dW$ be any local coloring.
            If $i>0$, let $\chi_\dW=\chi_{i-1,j}$.
            Let $v$ be any vertex on the west boundary of $\dG_{i,j}$.
            In $T_{i,j}$, vertex $v$ matches $\chi_\dW$ in each constrained direction~$\mathcal{D}$ of~$v$.
        \end{lemma}
        \begin{proof}
            In $U_{\dG_{i,j}}^{\chi_\dW}$, vertex $v$ matches $\chi_\dW$ in direction $\mathcal{D}$ by definition.
            We need to show that also in $T_{i,j}$, vertex $v$ matches $\chi_\dW$ in direction $\mathcal{D}$.
            Each edge of $T$ incident to $v$ in direction~$\mathcal{D}$ lies in $F$, or $T_{i,j}$, or a triangulation $T_{i',j'}$ with $(i',j')\neq(i,j)$ and $i'\leq i$ and $j'\leq j$.
            Let \[H=F\cup U_{\dG_{i,j}}^{\chi_\dW}\cup\bigcup \{T_{i',j'} \mid \text{$(i',j')\neq(i,j)$ and $i'\leq i$ and $j'\leq j$}\}.\]
            
            By inspection of the gadgets constructed via Figure~\ref{fig:grid_gadgets}, observe that in $H$, the cardinal degree of $v$ in direction $\mathcal{D}$ is as prescribed by $\sigma$.
            Because $T$ is a realization of $\sigma$, the vertex $v$ in direction $\mathcal{D}$ has the same cardinal degrees in $T$ as it does in $H$.
            Therefore, the cardinal degree of $v$ in direction $\mathcal{D}$ in $T_{i,j}$ is the same as that in $U_{\dG_{i,j}}^{\chi_\dW}$.
            Thus, also in $T_{i,j}$, vertex $v$ matches $\chi_\dW$ in direction $\mathcal{D}$.
        \end{proof}

        Lemma~\ref{lem:south} analogously concerns Condition~\ref{cond:valid_south} and can be proven symmetrically.
        
        \begin{lemma}[Condition~\ref{cond:valid_south} Assuming Intended Triangulations to the South-West]
        \label{lem:south}
            Let~$T$ be a cardinal signature realization of $\sigma$ and $T_{i,j}\subseteq T$ be its triangulation of gadget $\dG_{i,j}$.
            Fix some~$i$ and $j$ and assume for all $(i',j')\neq (i,j)$ with $i'\leq i$ and $j'\leq j$, that there exists a local coloring $\chi_{i',j'}$ such that $T_{i',j'}=U_{\dG_{i',j'}}^{\chi_{i',j'}}$.

            If $i=0$, let $\chi_\dS$ be any local coloring.
            If $i>0$, let $\chi_\dS=\chi_{i,j-1}$.
            Let $v$ be any vertex on the south boundary of $\dG_{i,j}$.
            In $T_{i,j}$, vertex $v$ matches $\chi_\dS$ in each constrained direction~$\mathcal{D}$ of~$v$.
        \end{lemma}

        Lemmas~\ref{lem:interiornorthandeast}--\ref{lem:south} together establish Conditions~\ref{cond:valid_vertex}--\ref{cond:valid_south} for a gadget assuming that all gadgets to its south and west use intended triangulations.
        Lemma~\ref{lem:induction} follows by induction, and we deduce that every realization of $\sigma$ uses an intended realization for each gadget.
    \begin{restatable}{lemma}{induction}
        \label{lem:induction}
            If all gadgets are correct, then all cardinal signature realizations of $\sigma$ lie in~$\mathcal{R}$.
        \end{restatable}
        \begin{proof}
            Let $T$ be a cardinal signature realization of $\sigma$.
            Let $\dG_{i,j}$ denote the gadget in column~$i$ from the left and row $j$ from the bottom, and let~$T_{i,j}$ denote its triangulation in~$T$.
            To establish that $T$ lies in $\mathcal{R}$, it remains to show that each $T_{i,j}$ is an intended realization of $\dG_{i,j}$.

            We use strong induction, fixing some $i$ and $j$ and assuming, for all $(i',j')\neq(i,j)$ with~$i'\leq i$ and $j'\leq j$, 
            that $T$ triangulates $\dG_{i',j'}$ with one of its intended realizations.
            It suffices to prove that $T$ triangulates $\dG_{i,j}$ using one of its intended realizations.
            By Lemmas~\ref{lem:west} and~\ref{lem:south} respectively, $T_{i,j}$ satisfies Conditions~\ref{cond:valid_west} and~\ref{cond:valid_south} of Definition~\ref{def:valid_triangulation}.
            Combining this with Lemma~\ref{lem:interiornorthandeast},~$T_{i,j}$ also satisfies the remaining Condition~\ref{cond:valid_vertex} of Definition~\ref{def:valid_triangulation}, so $T_{i,j}$ is valid.
            Because we assumed that $\dG_{i,j}$ is correct, $T_{i,j}$ is an intended realization of~$\dG_{i,j}$.
        \end{proof}

\subsection{Tools for Correctness of Gadgets}
    The proofs of correctness for our gadgets utilize various arguments, but
    several arguments appear repeatedly. We establish these most frequently used
    arguments in this section. 
    
    \begin{lemma}[Convex Path Argument]
        \label{lem:convexPath}
        Let $G$ be a PSL triangulation containing
        the cycle $C={v_1, v_2, \ldots v_n}$.
        Furthermore, suppose that for all $2 \leq i< j \leq n$, the segment
        $v_iv_j$ lies outside of $C$.
        Then $G$ contains the edges $v_1 v_i$ for all $2 \leq i \leq n$ (see Figure~\ref{fig:convex-path-argument}).
        
        \begin{figure}[h]
            \centering
            \includegraphics{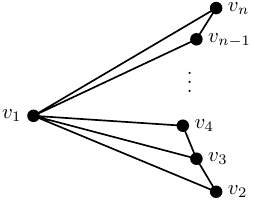}
            \caption{Illustration of the framework for Lemma~\ref{lem:convexPath}.}
            \label{fig:convex-path-argument}
        \end{figure}
    \end{lemma}
    
    \begin{proof}
         Because $C$ contains no interior vertices, any triangulation of $C$ has $n-3$ interior edges.
         Since segments $v_iv_j$ lie outside of $C$, all edges used to triangulate the interior of~$C$ are of the form $v_1v_i$ for $2 < i < n$.
         There are exactly $n-3$ such edges, so $G$ contains them~all.
    \end{proof}
    
    \begin{lemma}[Ear Clipping Argument]
        \label{lem:ear}
       Let $G$ be a PSL triangulation containing edges $v_1v_2$ and $v_2v_3$. 
       Suppose that the triangle defined by $v_1,v_2$, and $v_3$ does not contain vertices in its interior.
       Suppose that $v_2$ is not adjacent to any edges that pass through $v_1v_3$.
       Then $G$ contains the edge $v_1v_3$.
    \end{lemma}
    
    \begin{proof}
       Suppose not. Without loss of generality, edge $v_1v_2$ must be part of
       two triangles that are not $v_1v_2v_3$, one of which forces an edge to
       $v_2$ that passes through $v_1v_3$, a contradiction.
    \end{proof}
    
    \begin{lemma}[Triangle Apex Argument]\label{lem:triangle-apex}
        Let $G$ be a PSL graph with vertex set $v_1, v_2, u_1, u_2, \ldots, u_k$.
        Suppose that for all $i<j$, the segments $v_1u_j$ and $v_2u_i$ cross.
        Suppose that $v_1$ and $v_2$ connect to at least $m_1$ and $m_2$ of the vertices $u_1,\dots,u_k$, respectively.
        Then $v_1$ connects to at least $m_1$ of the vertices $u_1,\dots,u_{k+1-m_2}$ and $v_2$ connects to at least $m_2$ of the vertices $u_{k+1-m_1},\dots,u_k$.
        In particular, if $m_1+m_2=k+1$, then the edges are completely determined.
    \end{lemma}
    \begin{proof}
        Otherwise, $v_1$ connects to some $u_j$ and $v_2$ connects to some $u_i$ with $i<j$, which causes a crossing.
    \end{proof}

    We often encounter $k=2,3$ when invoking Lemma~\ref{lem:triangle-apex}, so
    we describe these small and specific instances separately in the following
    corollary.
    
    \begin{corollary}[Non-Crossing Subgraphs of $K_{2,2}$ and $K_{2,3}$]
    \label{cor:polygon}
        Let $G$ be a PSL graph with vertex set $v_1,v_2,u_1,u_2$ (respectively $v_1,v_2,u_1,u_2,u_3$), such that the segments between $v_i$ and~$u_j$ cross if and only if they cross in Figure~\ref{fig:small_polygon} (a), respectively (b).
        Suppose that $G$ has at least three (respectively four) edges of the illustrated edges.
        Then the figure illustrates all possible cases for the edges of $G$.
    \end{corollary}

    \begin{figure}[h]
        \centering
        \includegraphics{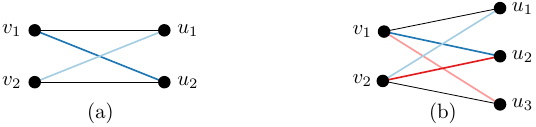}
        \caption{Given vertex sets as shown, we see all possible cases for edge-maximal non-crossing subgraphs of $K_{2,2}$ and $K_{2,3}$, respectively. Namely, in (a), we may choose either light or dark blue. In (b), we may choose dark blue and dark red, dark blue and dark red, or light blue and light red. In either case, we always must include the black edges.}
        \label{fig:small_polygon}
    \end{figure}
    
    \begin{lemma}[Almost-Complete Fan Argument]
    \label{lem:fan}
        Let $G$ be a PSL graph on the vertices $v_1$, $v_2$, $u_1,u_2,\dots,u_k$, such that $v_1$ is adjacent to four out of $v_2,u_1, u_2, \ldots, u_k$, and $v_2$ is adjacent to $k-1$ (i.e., all but one) of $u_1, \ldots, u_k$.
        Suppose that segment $v_1u_2$ crosses $v_2u_1$, and $v_1u_{k-1}$ crosses $v_2u_k$.
        Suppose also that each segment $v_1u_i$ with $3 \leq i \leq k-2$ crosses either both $v_2u_1$ and $v_2u_2$ or both $v_2u_k$ and $v_2u_{k-1}$.
        Then in $G$ (as depicted in Figure~\ref{fig:fan-lemma}):
        \begin{enumerate}
            \item $v_1$ is adjacent to $v_2$, $u_1$, $u_k$, and either (a) $u_2$ or (b) $u_{k-1}$,
            \item $v_2$ is adjacent to $u_i$ for all $2 \leq i \leq k-1$, and in case (a) $u_k$ or in case (b) $u_1$.
        \end{enumerate}
        \begin{figure}[h]
            \centering
            \includegraphics{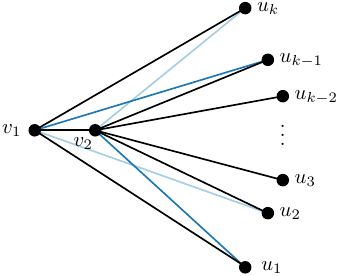}
            \caption{Illustration of the framework for Lemma~\ref{lem:fan}. Cases (a) and (b) are indicated by light and dark blue, respectively.}
            \label{fig:fan-lemma}
        \end{figure}
    \end{lemma}
    
    \begin{proof}
        If $v_1$ is adjacent to any $u_i$ with $3\leq i\leq k-2$, or $v_1$ is simultaneously adjacent to both $u_2$ and $u_{k-1}$, then $v_2$ can be adjacent to at most $k-2$ of $u_1,\dots,u_k$, as edges would cross otherwise.
        Therefore $v_1$ is adjacent to four of $v_2$, $u_1$, $u_2$, $u_{k-1}$, and $u_k$, and not both of~$u_2$ and $u_{k-1}$, so \enumit{1.} follows.
        If $v_1$ is adjacent to (a) $u_2$, then $v_2$ cannot be adjacent to $u_1$, so it must be adjacent to each of $u_2,\dots,u_k$.
        If instead $v_1$ is adjacent to (b) $u_k$, then $v_2$ cannot be adjacent to $u_k$, so it must be adjacent to each of $u_1,\dots,u_{k-1}$.
        So also \enumit{2.}~follows.
    \end{proof}

\subsection{Correctness of Gadgets}\label{sec:gadget-correctness}


    \begin{figure}
        \centering
        \includegraphics[page=1]{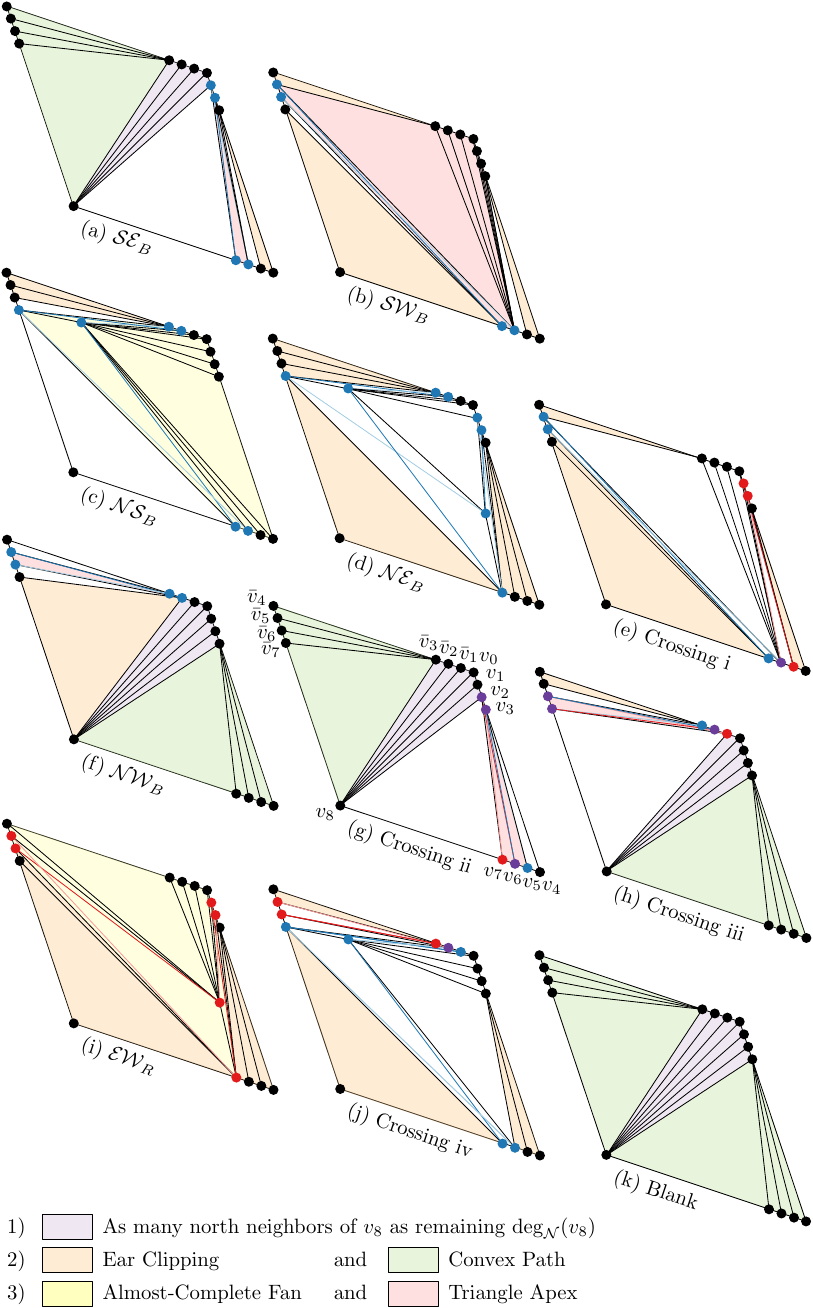}
        \caption{Visual guides for the proofs of Lemmas~\ref{lem:simple}--\ref{lem:crossingiv}. 
        Blue and red vertices are incident to
        blue and red edges, respectively, and purple vertices are incident to
        both blue and red edges. 
        The argument that forces an interior edge is indicated by the color of an incident triangle.
        Interior edges that are not incident to any colored triangle are handled separately.
        }
        \label{fig:sheep_chop}
    \end{figure}
    
    The proofs of correctness for many of our gadgets are nearly
    identical.
    Therefore, in the following lemma, we elect to provide the proof for one representative in
    greater detail, and trust the reader to extrapolate
    to others. Recall that triangulations of gadgets contain all edges of the frame cell boundary. In the correctness proofs of our gadgets, we therefore assume the boundary to be given, and for ease of exposition, we do not include boundary edges in the calculations of (cardinal) degrees. 
    We begin with the gadget Crossing~ii, as it illustrates a number of key ideas used in other cells.

    \begin{lemma}[Correctness of Crossing~ii]
        \label{lem:simple}
        The gadget Crossing~ii,
        given in
        Figure~\ref{fig:grid_gadgets}~(g) is correct.
    \end{lemma}
    
    \begin{proof}
        The proof steps are illustrated in Figure~\ref{fig:sheep_chop}(g). Let $T$ be a valid triangulation of the frame cell Crossing~ii. From Definition~\ref{def:valid_triangulation} we get certain conditions on vertices in $T$ matching local colorings that must be met. We use these conditions to reconstruct all valid triangulations, calling an edge \emph{forced} if it must be present in any valid triangulation $T$.

        First, observe that for $v_8$ to match a local coloring we must have $\deg(v_8) = 6$ (note that all cardinal directions are constrained for $v_8$), and for $v_3$ to match a local coloring we must have $\deg_\dW(v_3) = 0$ (note that the west direction is constrained for $v_3$). So in particular, $v_8$ cannot be adjacent to $v_3$. Convexity of the boundaries also implies that, within this frame cell, $v_8$ cannot be adjacent to any vertices on the south and west boundaries. Thus, the only way for $v_8$ to have degree six is to be adjacent to $\vb_3, \vb_2, \vb_1, v_0, v_1, v_2$, giving us the six corresponding forced edges (illustrated by the gray shaded region).

        From here we see that vertices $\vb_3, \vb_4, \vb_5, \vb_6, \vb_7,$ and
        $v_8$ form a cycle of forced edges, and for $4 \leq i < j \leq 8$, the
        segment $v_iv_j$ lies outside of that cycle. Then by
        Lemma~\ref{lem:convexPath} (the ``convex path argument''), we see that – for $T$ to be a triangulation –
        $\vb_3$ must be adjacent to~$\vb_5, \vb_6,$ and $\vb_7$ (illustrated by the green shaded region).

        Next, we take care of the pentagon $v_2v_3v_5v_6v_7$, which lies in the correct position for Corollary~\ref{cor:polygon}(b). In any local coloring, vertices $v_5, v_6, v_7$ have combined (north) degree four, and since edge $v_2v_8$ is forced by our first argument they must be adjacent to $v_2$ and/or $v_3$. I.e., four edges of the ones illustrated in Figure~\ref{fig:small_polygon} must be present in $T$. Then, by Corollary~\ref{cor:polygon}, edges $v_2v_7$ and $v_3v_5$ are forced, and $T$ contains either exactly the dark red and dark blue, the dark red and light blue, or the light red and dark blue edge, which are the three intended triangulations. Then by Definition~\ref{def:correct_gadget}, Crossing~ii as pictured in Figure~\ref{fig:sheep_chop}~(g) is correct.
    \end{proof}
    
    The correctness proofs for the gadgets $\dS\dE$, $\dS\dW$, $\dN\dS$, $\dN\dW$, Crossing~iii, $\dE\dW$, and the Blank gadget follow a
    similar structure, and Figure~\ref{fig:sheep_chop} contains a visual key for the logical steps we take, where such steps are taken in reading order as indicated by the key.
    Note that proofs of correctness for monotone tiles are symmetric for blue or red, so we omit subscripts $B$ or $R$ in proofs, and only illustrate one of the two scenarios in subfigures of Figure~\ref{fig:sheep_chop}.
    
    \begin{lemma}[Correctness of Remaining Simple Gadgets]
        \label{lem:simple}
        The gadgets $\dS\dE$, $\dS\dW$, $\dN\dS$, $\dN\dW$,
        Crossing~iii, $\dE\dW$, and the Blank gadget given in
        Figures~\ref{fig:sheep_chop}(a), (b), (c), (f), (h), (i), and (k),
        respectively, are correct.
    \end{lemma}
    
    We present the correctness of gadgets that require individualized arguments separately.

    \begin{restatable}[Correctness of $\dN\dE$]{lemma}{NE}
    \label{lem:NE}
             The $\dN\dE$ gadget given
             Figure~\ref{fig:sheep_chop}(d) is correct.
    \end{restatable}
    \begin{proof} 
        Let $T$ be a valid triangulation of $\dN\dE$. Vertex $v_8$ has degree 0 in all local colorings, and no vertices lie inside the triangle $\vb_7v_8v_7$. Thus, edge $\vb_7v_7$ is forced by the Ear Clipping argument (Lemma~\ref{lem:ear}).
        Similarly, edges from $\vb_3$ to $\vb_5, \vb_6,\vb_7$ are forced, and edges from $v_3$ to~$v_5,v_6,v_7$ are forced via Ear Clipping
        arguments made in that order. For the former three edges we use that the south degree is constrained for vertices $\vb_4, \vb_5$ and $\vb_6$, and for the latter three edges we use that the west degree is constrained for vertices $v_4, v_5$ and $v_6$.
        
        The west degree of $p$ is one in all local colorings, and the edge to $\vb_7$ is the only possibility
        that does not introduce an edge crossing, meaning $p\vb_7$ is a forced edge.
        Similarly, because of the south degree of $q$, edge $qv_7$ is forced.
        
        We have $\deg_\dN(v_2) = 0$ and $\deg_\dN(v_3) = 0$ in all local colorings, and the north direction is constrained for $v_2$ and $v_3$. In particular, neither $v_2$ nor $v_3$ can be incident to $p$.
        
        This allows us to use the Almost-Complete Fan argument (Lemma~\ref{lem:fan}) for $\vb_7,p$ (as $v_1,v_2$) and vertices $\vb_3, \vb_2, \vb_1, v_0, v_1, q, v_7$ (as $u_1,\dots,u_k$), yielding exactly the intended triangulations as pictured in the red area of Figure~\ref{fig:sheep_chop}(d). It remains to look at the pentagon $qv_7v_3v_2v_1$. By the Triangle Apex Argument (Lemma~\ref{lem:triangle-apex}) applied to $q$ and $v_7$ (as $v_1,v_2$) and $v_1,v_2,v_3$ (as $u_1,u_2,u_3$) the edges $v_1q$ and $v_2v_7$ are also forced, leaving a choice of either $v_1v_7$ or $v_2q$ (light blue or dark blue), which needs to be consistent with the other choices of blue edges, for otherwise $T$ would not match a local coloring in vertices $q$ and $v_7$. It follows that $T$ is one of the intended triangulations. 
    \end{proof}
    
    \begin{restatable}[Correctness of Crossing~i]{lemma}{crossingi}
    \label{lem:crossingi}
        Gadget Crossing~i given in Figure~\ref{fig:sheep_chop}(e) is correct.
    \end{restatable}
    \begin{proof}
        Let $T$ be a valid triangulation of Crossing~i. Edge $\vb_7v_7$ is
        forced by the same reasoning as in the proof of Lemma~\ref{lem:NE}. We can then make a similar argument to show edges $\vb_6v_7$,
        $\vb_3\vb_5$, and $v_3v_5$ are forced by the Ear Clipping argument (Lemma~\ref{lem:ear}), using the east degree of $\vb_6$, the south degree of $\vb_4$, and the west degree of $v_4$, respectively (note that the corresponding directions are each constrained for the respective vertex).

        Note that the south direction is constrained for $\vb_5, \vb_3,$ and $\vb_2$, and the north direction is constrained for $v_6$. In all local colorings, we have $\deg_{\dS}(\vb_5) \geq 1$, $\deg_{\dS}(\vb_3) = 1$, $\deg_{\dS}(\vb_2) = 1$, and $\deg_{\dN}(v_6) \geq 6$, thus $T$ must satisfy these conditions. In particular, $\vb_5$ cannot be adjacent to $\vb_2$ or $\vb_1$ for $\vb_3$ to have sufficient south degree, it cannot be adjacent to $\vb_0$ for $\vb_2$ to have sufficient south degree, and it cannot be adjacent to $v_1, v_2, v_3$, and $v_5$ for $v_6$ to have sufficient north degree. I.e., $\vb_5$ can only be adjacent to $v_6$ and $v_7$ in $T$, and at least one of the two edges exists, implying that $\vb_6$ can also only be adjacent to $v_6$ and $v_7$ in $T$. In all colorings, vertices $\vb_5$ and $\vb_6$ have remaining combined degree 2, thus there must be two additional edges between $\vb_5, \vb_6$, and $v_6, v_7$. It follows from Corollary~\ref{cor:polygon}(b) that the edge $\vb_5v_6$ is forced, and exactly one of the two edges $\vb_5v_7$ and $\vb_6v_6$ is in $T$.

        Using remaining east degrees of vertices $\vb_5, \vb_3, \vb_2, \vb_1$ we get that the edges $\vb_3v_6$, $\vb_2v_6$, $\vb_1v_6$, and $v_0v_6$. It remains to investigate what $T$ can look like inside the forced convex cycle $v_5, v_6, v_0, v_1, v_2, v_3$. In all local colorings, vertices $v_5$ and $v_6$ have remaining north degrees at least 1. I.e., together with the forced edge $v_0v_6$ vertex $v_6$ is adjacent to at least two out of the four vertices $v_0, v_1, v_2, v_3$, and together with the forced edge $v_3v_5$ vertex $v_5$ is adjacent to at least two out of the four vertices $v_0, v_1, v_2, v_3$. Thus, by the Triangle Apex argument (Lemma~\ref{lem:triangle-apex}), $v_1v_6$ and $v_2v_5$ are also forced, leaving the remaining edge to be either $v_1v_5$ or $v_2v_6$. We have thereby shown that for $T$ to be valid, it has to be one of the intended triangulations.
    \end{proof}
    
    \begin{restatable}[Correctness of Crossing~iv]{lemma}{crossingiv}
    \label{lem:crossingiv}
        Gadget Crossing~iv given in Figure~\ref{fig:sheep_chop}(j) is correct.
    \end{restatable}
    \begin{proof}
        Let $T$ be a valid triangulation of Crossing~iv. By Ear Clipping arguments (Lemma~\ref{lem:ear}) we get forced edges $v_3v_5$, $v_3v_5$, $\vb_7v_7$, and $\vb_3\vb_5$ in a similar manner as in previous proofs. 
        
        In all local colorings we have $\deg_\dE(p) = 7$, and the vertices $\vb_5, \vb_6, \vb_7$ have combined remaining east degree 5. The former implies that triangulation $T$ must satisfy that $p$ is adjacent to seven out of the nine vertices $\vb_3, \vb_2, \vb_1, v_0, v_1, v_2, v_3, v_6, v_7$. We claim that these seven vertices must be contiguous. Let $a$ be the first among those vertices (in the order they are listed) that $p$ is adjacent to, and $b$ the last. The remaining 5 edges adjacent from the east to $\vb_5, \vb_6, \vb_7$ must then be inner edges in the polygon $\vb_7, \vb_6, \vb_5,\vb_3, \dots, a,p, b, \dots,v_7$. A polygon with $n$ vertices can have at most $n-3$ non-crossing inner edges. Conversely, the polygon described above must have at least 8 vertices. This implies the correctness of our claim, implying in turn that $(a,b)$ is either $(\vb_3,v_3)$, $(\vb_2, v_6)$, or $(\vb_1, v_7)$. We will rule out the case $(\vb_3, v_3)$. So for the sake of contradiction, assume that $(a,b) = (\vb_3,v_3)$. Since $p$ has west degree one in all local colorings, it has to be adjacent to some vertex $u \in \{\vb_5, \vb_6, \vb_7\}$. But since $p$ has no remaining degree, and since the triangle $pua$ contains no other vertices in all cases, we will also have the edge $ua$ by the Ear Clipping argument (Lemma~\ref{lem:ear}), and similarly the edge $ub$. Since vertices $\vb_5$ and $\vb_6$ each have remaining degree at most two in all local colorings, $u$ must be $\vb_7$. However, edge $\vb_7\vb_3$ allows only for one additional edge in the quadrilateral $\vb_3\vb_5\vb_6\vb_7$, which is a contradiction to $\vb_5$ and $\vb_6$ having remaining total degree 3 in all local colorings. We can now apply the Almost-Complete Fan Argument (Lemma~\ref{lem:fan}) to vertices $\vb_7$ and $p$ as $v_1$ and $v_2$ and vertices $\vb_2, \vb_1, v_0, v_1, v_2, v_3, v_6, v_7$ as $u_1,\dots,u_k$. This gives the forced edge $\vb_2\vb_7$ and inside the forced polygon $\vb_7\vb_2\vb_1v_0v_1v_2v_3v_6v_7$ exactly the intended triangulations. The remaining part follows from the application of Corollary~\ref{cor:polygon}(a) to the vertices $\vb_2, \vb_3, \vb_5, \vb_6$.
    \end{proof}

    By combining the bijection in Theorem~\ref{thm:tile-selection-bijection} with Lemma~\ref{lem:induction} and the correctness of all gadgets, we conclude that the number of cardinal signature realizations of~$\sigma$ is the same as the number of noncrossing tile selections of $X$.
    This establishes Theorem~\ref{thm:sharpPhard}.

    \begin{theorem}\label{thm:sharpPhard}
        \sharpCardSig is \#P-hard.
    \end{theorem}

    \section{Discussion}
    We showed that \sharpCardSig, or counting the number of triangulations satisfying given degree information in each of the four cardinal directions is \#P-hard, via a reduction from \sharpVC through the auxiliary problem of \sharpTiles. Because we can shear the graphs used in this reduction so each vertex has neighbors only in antipodal wedges of arbitrarily small angle, our work has a straightforward extension to the problem of counting PSL triangulations when we know degree information in any $d$ directions, with $d \geq 4$.
    In ongoing work, we are interested in the hardness of determining whether or not a PSL triangulation that satisfies a given cardinal degree signature exists.

    In terms of our original motivation from topological data analysis, namely studying the inverse problems for directional transforms, we note that this leaves a great deal of open questions and room for future progress.  We have focused on triangulations in $\mathbb{R}^2$, as well as data given from the directional transform with Euler characteristic or persistence data, but other transforms such as the radial or affine Grassmannian have not been as well studied~\cite{Onus2024,Chambers2025}. One could also study different topological signatures, such as mapper graphs or merge trees, and frame inverse questions for low dimensional structures.
    
    \bibliography{references.bib}

\end{document}